\newcommand{\ifArxivVersion}{\iftrue}
\newcommand{\ifLocalVersion}{\iffalse}

\ifArxivVersion  
    \newcommand{\algseparator}{\;}
\else
	\newcommand{\algseparator}{}
\fi

\documentclass{article}[11pt,a4paper]
\pdfoutput=1
 
\usepackage{microtype}
\usepackage{amsmath}
\usepackage{amssymb}
\usepackage{amsthm}
\usepackage[symbol*]{footmisc}
\usepackage{authblk}
\usepackage{cite}
\usepackage{fullpage}
\usepackage{hyperref}

\newtheorem{definition}{Definition}[section]
\newtheorem{lemma}[definition]{Lemma}
\newtheorem{theorem}[definition]{Theorem}

\newcommand{\bigo}{\mathcal{O}}

\newcommand{\ballcolor}{\mathsf{color}}
\newcommand{\balance}{\mathsf{balance}}

\newcommand{\predict}{\mathsf{predict}}
\newcommand{\cmp}{\mathsf{cmp}}
\newcommand{\E}{\mathbb{E}}
\newcommand{\cnt}{\mathsf{cnt}}

\newcommand{\cupdot}{\mathbin{\mathaccent\cdot\cup}}

\usepackage[noend, noline, ruled, linesnumbered]{algorithm2e}
\SetAlFnt{\normalsize}
\DontPrintSemicolon
\SetKwComment{Comment}{$\triangleright$\ }{}

\hypersetup{
    colorlinks=true,
    linkcolor=black,
    citecolor=black,
    filecolor=black,
    urlcolor=black,
    pdfauthor={Pawe\l{} Gawrychowski, Jukka Suomela, Przemys\l{}aw Uzna\'nski},
    pdftitle={Randomized algorithms for finding a majority element}
}

\title{Randomized algorithms for finding a majority element}

\author[1]{Pawe\l{} Gawrychowski}
\author[2]{Jukka Suomela}
\author[3]{Przemys\l{}aw~Uzna\'nski}
\affil[1]{Institute of Informatics, University of Warsaw, Poland}
\affil[2]{Helsinki Institute for Information Technology HIIT, \mbox{Department of Computer Science, Aalto University, Finland}}
\affil[3]{Department of Computer Science, ETH Z\"urich, Switzerland}

\date{}    

\begin{document}

\maketitle

\begin{abstract}
Given $n$ colored balls, we want to detect if more than $\lfloor n/2\rfloor$ of them
have the same color, and if so find one ball with such majority color. We are only allowed
to choose two balls and compare their colors, and the goal is to minimize
the total number of such operations. A well-known exercise is to show how to find such
a ball with only $2n$ comparisons while using only a logarithmic number of bits for
bookkeeping. The resulting algorithm is called the Boyer--Moore
majority vote algorithm. It is known that any deterministic method needs
$\lceil 3n/2\rceil-2$ comparisons in the worst case, and this is tight.
However, it is not clear what is the required
number of comparisons if we allow randomization. We construct a randomized
algorithm which always correctly finds a ball of the majority color (or detects that there is
none) using, with high probability, only $7n/6+o(n)$ comparisons.
We also prove that the expected number of comparisons used by any such randomized
method is at least $1.019n$.
\end{abstract}

\section{Introduction}

A classic exercise in undergraduate algorithms courses is to construct a linear-time
constant-space algorithm for finding the majority in a sequence of $n$ numbers
$a_1,a_2,\ldots,a_n$, that is, a number $x$ such that more than $\lfloor n/2\rfloor$
numbers $a_i$ are equal to $x$, or detect that there is no such $x$. The solution
is to sweep the sequence from left to right while maintaining a candidate and a counter.
Whenever the next number is the same as the candidate, we increase the counter;
otherwise we decrease the counter and, if it drops down to zero, set the candidate to
be the next number. It is not difficult to see that if the majority exists, then it is
equal to the candidate after the whole sweep, therefore we only need to count how
many times the candidate occurs in the sequence. This simple yet beautiful solution
was first discovered by Boyer and Moore in 1980; see~\cite{MJRTY} for the history of the problem.

The only operation on the input numbers used by the Boyer--Moore algorithm is testing
two numbers for equality, and furthermore at most $2n$ such checks are ever being made.
This suggests that the natural way to think about the algorithm is that the input consists
of $n$ colored balls and the only possible operation is comparing the colors of any two
balls. Now the obvious question is how many such comparisons are necessary and
sufficient in the worst possible case. Fischer and Salzberg~\cite{FischerSalzberg} proved that the answer is
$\lceil 3n/2\rceil-2$. Their algorithm is a clever modification of
the original Boyer--Moore algorithm that reuses the results of some previously made
comparisons during the verification phase. They also show that no better solution
exists by an adversary-based argument. However, this argument assumes that the
strategy is deterministic, so the next step is to allow randomization.

Surprisingly, not much seems to be known about randomized algorithms for computing
the majority in the general case. For the special case of only two colors, 
Christofides~\cite{Christofides} gives a randomized algorithm that uses $\frac{2}{3}(1-\frac{\epsilon}{3})n$
comparisons in expectation and returns the correct answer with probability $1-\epsilon$,
and he also proves that this is essentially tight; this improves on
a previous lower bound of $\Omega(n)$ by De Marco and Pelc~\cite{MarcoPelc}.
Note that in the two-color case any deterministic algorithm needs precisely $n-B(n)$ comparisons, where $B(n)$
is the number of 1s in the binary expansion of $n$, and this is 
tight~\cite{SaksWerman,AlonsoLowerbound,Wiener}. For a random input, with each
ball declared to be red or blue uniformly at random, roughly $2n/3$ comparisons
are sufficient and necessary in expectation to find the majority color~\cite{AlonsoAverage}.
To the best of our knowledge upper and lower bounds on the expected number
of comparisons with unbounded number of colors have not been studied
before.

Related work include \emph{oblivious} algorithms studied by Chung et al.~\cite{oblivious},
that is, algorithms in which subsequent comparisons do not depend on the previous answers,
and finding majority with larger queries~\cite{MarcoK15,Eppstein,VizerGKPPW15}.
Another generalization is finding a ball of plurality color, that is, the color that occurs
more often than any other~\cite{AignerMM05,GerbnerKPP13,KralST08}.

We consider minimizing the number of comparisons mostly as an academic exercise,
and believe that a problem with such a simple formulation deserves to be thoroughly
studied. However, it is possible that a single comparison is so expensive that their number
is the bottleneck. Such a line of thought motivated a large body of work
studying the related questions of the smallest number of comparisons required to find
the median element; see~\cite{DorZ99,DorZ01,paterson1996progress} and the references
therein. Of course, the simplest Boyer--Moore algorithm has the advantage of
using only two sequential
scans over the input and a logarithmic number of bits, while our algorithm
needs more space
and random access to the input.

Given that the original motivation of Boyer and Moore was fault-tolerant computing,
we find it natural to consider Las Vegas algorithms,
that is, the number of comparisons depends on the random choices of the algorithm
but the answer is always correct. This way the result is correct even if the source of random
bits is compromised; an adversary controlling the random number generator can only
influence the running time.

\paragraph{Model.} We identify balls with numbers $1,2,\ldots,n$. We write $\cmp(i,j)$ for
the result of comparing the colors of balls $i$ and $j$ (true for equality, false for inequality). We consider randomized
algorithm that, after performing a number of such comparisons, either finds a ball
of the majority color or detects that there is no such color. A majority color is a color 
with the property that more than $\lfloor n/2\rfloor$ balls are of such color. The algorithm
should always be correct, irrespectively of the random choices made during the execution.
However, the colors of the balls are assumed to be fixed in advance, and therefore
the number of comparisons is a random variable. We are interested in minimizing
its expectation.

\paragraph{Contributions.} We construct a randomized algorithm, which always correctly
determines a ball of the majority color or detects that there is none, using $7n/6+o(n)$
comparisons with high probability (in particular, in expectation). We also show that
the expected number of comparisons used by any such algorithm must be at least
$1.019n$.

\section{Preliminaries}

We denote the set of balls (items) by $M = \{1,2,\ldots,n\}$. We write $\ballcolor(x)$ for the color of
ball $x$, and $\cmp(x,y)$ returns true if the colors of balls $x$ and $y$ are identical.

An event occurs \emph{with very high probability} (w.v.h.p.)\ if it happens with probability at least
$1 - \exp(-\Omega(\log^2 n))$. Observe that the intersection of polynomially many very high
probability events also happens with very high probability. 

\begin{lemma}[symmetric Chernoff bound]
The number of successes for $n$ independent coin flips is w.v.h.p.\ at most  $\frac{n}2+\bigo(\sqrt{n} \log n)$.
\end{lemma}

\begin{lemma}[sampling]
\label{lem:sampling}
Let $X \subseteq M$ such that $|X| = m$. Let $m'$ denote the number of hits on
elements from $X$ if we sample uniformly at random $k \le n$ elements from $M$
without replacement. Then w.v.h.p.\ $\left|m'/k - m/n\right| = \bigo(k^{-1/2} \log n)$.
\end{lemma}

\begin{proof}
Let $\delta = \Theta( k^{-1/2} \log n)$.
From the tail bound for hypergeometric distribution (see Chv\'atal~\cite{chvatal}), we have
\[\Pr\Bigl( \Bigl|\frac{m'}{k} - \frac{m}{n}\Bigr| \ge \delta\Bigr) \le 2 \exp( -2 \delta^2 k) =  \exp( -\Theta(\log^2 n)).\]

Worth noting is that the error bound here would be essentially the same if we replace ``without replacement'' with ``with replacement''. In that case, we would be invoking a tail bound for binomial distribution.
\end{proof}

Now we consider a process of pairing the items without replacement (choosing a random
perfect matching on $M$; if $n$ is odd then one item remains unpaired). For any $X \subseteq M$,
let $u_{XX}$ be a random variable counting the pairs with both elements belonging to $X$ when
choosing uniformly at random $\frac{n}{2}$ pairs of elements from $M$ without replacement.
Of course $\E[u_{XX}] = \frac{|X|(|X|-1)}{2(n-1)}$.

\begin{lemma}[concentration for pairs]
\label{lem:pairs}
For any $X\subseteq M$ w.v.h.p.
\[\left| u_{XX} - \frac{|X|^2}{2n} \right| = \bigo\bigl(\sqrt{|X|} \log n\bigr).\]
\end{lemma}

\begin{proof}
For simplicity, we can assume that $n$ is even.
Set $\Delta = \Theta(\sqrt{|X|} \log n)$.
Instead of choosing a random perfect matching on $M$, we consider an equivalent two-step random process:
\begin{enumerate} 
\item We choose uniformly at random a partition of $M$ into $M_1,M_2$ such that $|M_1|=|M_2|=\frac{n}{2}$. This partition induces a partition of $X$ into $X_1,X_2$, such that $X_1 = M_1 \cap X$ and $X_2 = M_2 \cap X$.

Observe that, by Lemma~\ref{lem:sampling}, $|X_1|,|X_2| \in \bigl[\frac{1}{2}|X|-\Delta,\, \frac{1}{2}|X|+\Delta\bigr]$ w.v.h.p.

\item Now, instead of pairing uniformly at random the elements from $M_1$ with the elements from $M_2$, for our purposes it is enough to choose uniformly at random a set $M'_2 \subseteq M_2$ of elements that are paired with $X_1$, and count how many elements of $X_2$ we have chosen. Thus $u_{XX} = |M'_2 \cap X_2|$.

By Lemma~\ref{lem:sampling} we have
\[ \left |\frac{|M'_2 \cap X_2|}{|X_2|} - \frac{|M'_2|}{|M_2|}\right| = \bigo(|X_2|^{-1/2} \log |M_2|) \quad {w.v.h.p.},\]
so the following holds (since $|M'_2| = |X_1|$) w.v.h.p.:
\[\left|u_{XX} - \frac{|X_1||X_2|}{n/2}\right| = \bigo(\sqrt{|X_2|} \log |M_2|).\]
Hence
\[\left|u_{XX} - \frac{|X|^2}{2n}\right| = \bigo(\sqrt{|X_2|} \log |M_2|) + \bigo\biggl(\frac{\Delta^2}{n/2}\biggr) = \bigo(\Delta).\qedhere\]
\end{enumerate}
\end{proof}

\begin{lemma}[pairs in partition]
\label{lem:pairs2}
Let $\mathcal{F} = \{X_1,\ldots,X_m\}$ be a partition of $M$.  Then w.v.h.p.
\[\left|\sum_{X \in \mathcal{F}} u_{XX} - \sum_{X \in \mathcal{F}} \frac{|X|^2}{2n}\right| = \bigo(n^{2/3}\log n).\]
\end{lemma}

\begin{proof}
Let $\mathcal{F}'$ be a partition $\{X'_1,\ldots,{X'}_{m'}\}$ obtained from $\mathcal{F}$ by merging all sets $X_i$ such that $|X_i| < n^{2/3}$ into larger sets of size between $n^{2/3}$ and $2n^{2/3}$ (this can be done in a greedy fashion). Let $\mathcal{A}$ be the family of original large sets, $\mathcal{B}$ the family of original small sets, and $\mathcal{C}$ the family of new larger sets.

Since $\mathcal{B}$ is a finer partition than $\mathcal{C}$, we have that
\[\sum_{X \in \mathcal{B}} u_{XX} \le \sum_{X \in \mathcal{C}} u_{XX},\quad
\sum_{X \in \mathcal{B}} \frac{|X|^2}{2n} \le \sum_{X \in \mathcal{C}} \frac{|X|^2}{2n}.\]
Since all sets in $\mathcal{C}$ are smaller than $2n^{2/3}$, we obtain
\[\sum_{X \in \mathcal{C}} \frac{|X|^2}{2n} \le \biggl(\sum_{X \in \mathcal{C}} |X|\biggr) n^{-1/3} \le n^{2/3}.\]
Then, because all sets in $\mathcal{C}$ are large, by a a direct application of Lemma~\ref{lem:pairs}, w.v.h.p.
\[\sum_{X \in \mathcal{C}} u_{XX} \le \sum_{X \in \mathcal{C}} \left(\frac{|X|^2}{2n} + \bigo(\sqrt{|X|}\log n)\right) \le  n^{2/3} + \bigo(\log n) \sum_{X\in \mathcal{C}} \sqrt{|X|} \leq \bigo(n^{2/3}\log n).\]
Similarly, because all sets in $\mathcal{F}'$ are large, by an application of Lemma~\ref{lem:pairs}, w.v.h.p.
\[\left|\sum_{X \in \mathcal{F'}} u_{XX} - \sum_{X \in \mathcal{F'}} \frac{|X|^2}{2n}\right| \le \sum_{X \in \mathcal{F'}} \bigo\bigl(\sqrt{|X|} \log n\bigr)  = \bigo(n^{2/3} \log n).\]
Now our goal is to bound
\[\left|\sum_{X \in \mathcal{F}} u_{XX} - \sum_{X \in \mathcal{F}} \frac{|X|^2}{2n}\right|\]
which, because $|a+b| \leq |a|+|b|$, is at most
\[\left|\sum_{X \in \mathcal{B}} u_{XX} - \sum_{X \in \mathcal{C}} u_{XX}\right| + \left|\sum_{X \in \mathcal{F'}} u_{XX} - \sum_{X \in \mathcal{F'}} \frac{|X|^2}{2n}\right| + \left|\sum_{X \in \mathcal{B}} \frac{|X|^2}{2n} - \sum_{X \in \mathcal{C}} \frac{|X|^2}{2n}\right|.\]
The first and the third addend can be bounded by $\bigo(n^{2/3}\log n)$ because if $0\leq a\leq b$ then $|b-a|\leq b$, hence by plugging in the previously derived bounds we obtain
\[\left|\sum_{X \in \mathcal{F}} u_{XX} - \sum_{X \in \mathcal{F}} \frac{|X|^2}{2n}\right| \leq \bigo(n^{2/3}\log n) + \bigo(n^{2/3}\log n) + \bigo(n^{2/3}) = \bigo(n^{2/3}\log n).\qedhere\]
\end{proof}

\begin{lemma}
\label{lem:drawing_until}
Let $X \subseteq M$ such that $|X|=m$. Let $k(m',m,n)$ denote the number of draws without replacement until we hit $m'$ elements from $X$. Then w.v.h.p.
\[k(m',m,n) \le \frac{n}{m}m' + \bigo\Bigl(\frac{n}{\sqrt{m}} \cdot \log n\Bigr).\]
\end{lemma}

\begin{proof}
Denoting by $S$ the number hits on $X$ when drawing without replacement $k$ items, then
by Lemma~\ref{lem:sampling} w.v.h.p.
\[S \ge \frac{m}{n}k - C \cdot m^{1/2} \log n\]
for some constant $C>0$. Substituting $k = \frac{n}{m}m' + C\cdot \frac{n}{\sqrt{m}} \log n$ we obtain
\[S  \ge m' + C \cdot m^{1/2} \log n - C \cdot m^{1/2} \log n = m'\]
meaning that indeed w.v.h.p.\ after $k$ draws we will have at least $m'$ elements from $X$.
\end{proof}

\section{Algorithm}

In this section we describe a randomized algorithm for finding majority. Recall
that the algorithm is required to always either correctly determine a ball of the majority
color or decide that there is no such color, and the majority color is a color of more
than $\lfloor n/2\rfloor$ balls. For simplicity we will assume for the time being that
$n$ is even, as the algorithm can be adjusted for odd $n$ in a straightforward manner without any change to the asymptotic cost.
Hence to prove that there is a majority color, it is sufficient to find $n/2+1$ balls
with the same color. In such case our algorithm will actually calculate the multiplicity
of the majority color. To prove that there is no majority color, it is sufficient to partition
the input into $n/2$ pairs of balls with different colors.

\pagebreak

The algorithm consists of three parts. Intuitively, by choosing a small random sample
we can approximate the color frequencies and choose the right strategy:
\begin{enumerate}
\item[(i)] There is one color with a large
frequency. We use algorithm $\textsc{heavy}$. In essence, we have only one candidate for the majority, and
we compute the frequency of the candidate in a naive manner. If the frequency is too small, we need to form sufficiently many pairs of balls with different colors among the balls that are not of the candidate color. This can be done by virtually pairing the non-candidate color elements, and testing these pairs until we find enough of them that have distinct colors. Additionally, we show that one sweep through the pairs is enough.
\item[(ii)] There are two colors with frequencies close to $0.5$. Now we use algorithm $\textsc{balanced}$. In essence, we can now reduce the size of the input by a pairing process, and then find the majority recursively.
If the recursion finds the majority, the necessary verification step is speeded up by reusing the results of the comparisons used to form the pairs.
\item[(iii)] All frequencies are small. We use $\textsc{light}$ which, as $\textsc{balanced}$, applies pairing and recursion. However, if the recursive call reports the majority, we construct enough pairs with different colors: whenever we find a pair of elements with both colors different than the majority color found by the recursive call, we pair them with elements of the majority color. Here we speeded up the process by reusing the results of the comparisons used to form the pairs as well.
\end{enumerate}

We start with presenting the main procedure of the algorithm; see Algorithm~\ref{algo:majority}.
The parameters are chosen by setting $\alpha = \frac13$, $\varepsilon = n^{-1/10}$ and $\beta = 0.45$. In fact we could chose any $\beta \in (\beta_1,\beta_2)$, where $\beta_1 = 1-\frac1{\sqrt{3}} \approx 0.4226$ and $\beta_2 \approx 0.47580$ is a root to $p^3-19 p^2-8 p + 8 = 0$.

\begin{algorithm}
\caption{$\textsc{majority}(M)$\label{algo:majority}}
\lIf{$|M|=1$}{%
  \KwRet{$M[1]$} is the majority with multiplicity 1 in $M$
}
sample $M' \subseteq M$ such that $|M'| = n^{\alpha}$\;
let $v_1,v_2,\ldots,v_k$ be the representatives of the colors in $M'$\;
let $q_i |M'|$ be the frequency of $\ballcolor(v_i)$ in $M'$, where $q_1 \ge q_2 \ge \ldots \ge q_k$\label{line:sampling}\;
\If{$q_1,q_2 \in [\frac{1}{2}-4\varepsilon,\frac{1}{2}+4\varepsilon]$}{
  \KwRet{$\textsc{balanced}(M)$}
}\ElseIf{$q_1 \ge \beta$ and $q_1^2 \ge q_2^2+\ldots+q_k^2 + 2\varepsilon$}{
  \KwRet{$\textsc{heavy}(M,v_1)$}
}\Else{
  \KwRet{$\textsc{light}(M)$}
}
\end{algorithm}

Before we proceed to describe the subprocedures, we elaborate on the sampling performed
in line~\ref{line:sampling}. Intuitively, we would like to compute the frequencies of all colors
in $M$. This would be too expensive, so we select a small sample $M'$ and claim that the
frequencies of all colors in $M'$ are not too far from the frequencies of all colors in $M$.
Formally, let $p_1, p_2, p_3, \ldots, p_\ell$ be the frequencies of all colors in $M$, that is
there are $p_i \cdot n$ balls of color $i$ in $M$ and let $q_i$ be the
frequency of color $i$ in the sample $M'$. By Lemma~\ref{lem:sampling}, w.v.h.p.\ 
$|p_i - q_i| = \bigo(n^{-\alpha/2} \log n)=o(\varepsilon)$. We argue that $\sum_i q_i^2$
is a good estimation of $\sum_i p_i^2$.

\begin{lemma}
\label{lem:sum_of_squares}
Let $p_i$ be the frequency of color $i$ in $M$ and $q_i$ be its frequency in $M'$, where
$M'\subseteq M$ a random sample without replacement of size $n^\alpha$. Then w.v.h.p.
\[\left|\sum_i p_i^2 - \sum_i q_i^2\right| = \bigo(n^{-\alpha/3} \log n) = o(\varepsilon) .\]
\end{lemma}

\begin{proof}
Let $m  = n^{\alpha}$. We analyze the following two sampling methods. 

\begin{enumerate}
\item Partition the elements of $M$ into $\frac{n}{2}$ disjoint pairs uniformly at random.
Select $\frac{m}{2}$ of these pairs uniformly at random. Denote by $A_1$ and $A_2$
the pairs with both elements of the same colors in the first and the second pairing, respectively.
By Lemma~\ref{lem:pairs2}, w.v.h.p.\ $\left||A_1| - \frac{n}{2} \sum_i p_i^2\right| = \bigo(n^{2/3} \log n)$.
Observe that by Lemma~\ref{lem:sampling} w.v.h.p.\ $\left| |A_2| - \frac{m}{n} |A_1| \right| = \bigo(m^{1/2} \log n)$. Thus, by the triangle inequality, w.v.h.p.
\[ \left| \frac{|A_2|}{m/2} - \sum_i p_i^2\right| = \bigo(n^{-1/3} \log n) + \bigo(m^{-1/2} \log n).\]

\item Partition the elements of $M'$ into $\frac{m}{2}$ disjoint pairs uniformly at random, and denote
by $B$ all pairs with both elements of the same color. By Lemma~\ref{lem:pairs2}, w.v.h.p.\ 
$\left| |B| - \frac{m}{2} \sum_i q_i^2 \right| = \bigo(m^{2/3} \log n)$, or equivalently
$\left|\frac{|B|}{m/2} - \sum_i q_i^2 \right| = \bigo(m^{-1/3} \log n)$.
\end{enumerate}

Now, because $A_2$ and $B$ have identical distributions, by the triangle inequality we have
\[
\left|\sum_i p_i^2 - \sum_i q_i^2\right| = \bigo(n^{-1/3} \log n) + \bigo(m^{-1/2} \log n) + \bigo(m^{-1/3} \log n) = \bigo(m^{-1/3} \log n). \qedhere
\]
\end{proof}

Now we present the subprocedures; see Algorithms \ref{algo:balanced}--\ref{algo:light}.

\begin{algorithm}[p]
\caption{$\textsc{balanced}(M)$\label{algo:balanced}}
randomly shuffle $M$\;
$X \gets [], Y\gets []$\;
\For{$i=1$ \KwSty{to} $|M|/2$}{
   \eIf{$\cmp(M[2i-1],M[2i])$}{append $M[2i]$ to $X$}{append $M[2i-1]$ and $M[2i]$ to $Y$}
}
run $\textsc{majority}(X)$\;
\lIf{\text{there is no majority in $X$}}{\KwRet{}\,no majority in $M$}\algseparator
let $\ballcolor(v)$ be the majority with multiplicity $k$ in $X$\;
$\cnt \gets 2k$\;
\For{$i=1$ \KwSty{to} $|Y|/2$}{\label{line:counting1}
  \If{$\cmp(v,Y[2i-1])$}{
    $\cnt \gets \cnt+1$
  }\ElseIf{$\cmp(v,Y[2i])$}{
    $\cnt \gets \cnt+1$
  }
}
\eIf{$\cnt \leq |M|/2$}{\KwRet{no majority in $M$}}{\KwRet{$\ballcolor(v)$ is the majority with multiplicity $k$ in $X$}}
\end{algorithm}

\begin{algorithm}
\caption{$\textsc{heavy}(M,v)$\label{algo:heavy}}
$\cnt \gets 0$,  $X \gets []$\;
\For{$i=1$ \KwSty{to} $|M|$}{
  \If{$\cmp(v,M[i])$}{
    $\cnt \gets \cnt+1$
  }\Else{
    append $M[i]$ to $X$
  }
}
\lIf{$\cnt > |M|/2$}{%
  \KwRet{$\ballcolor(v)$ is the majority with multiplicity $k$ in $M$} \label{line:majority_return}
}\algseparator
$k \gets |M|/2-\cnt$\;
randomly shuffle $X$\;
\For{$i=1$ \KwSty{to} $|X|/2$}{ \label{line:heavy_loop}
  \lIf{$\neg\cmp(X[2i-1],X[2i])$}{$k \gets k-1$}
  \lIf{$k = 0$}{%
    \KwRet{no majority in $M$}
  }
}
\KwRet{\textsc{Boyer--Moore}(M)} \Comment*[r]{fallback, $2n$ comparisons}
\end{algorithm}
\begin{algorithm}
\caption{$\textsc{light}(M)$\label{algo:light}}
randomly shuffle $M$\;
$X \gets [], Y \gets []$\;
\For{$i=1$ \KwSty{to} $|M|/2$}{ \label{line:light_partition}
  \If{$\cmp(M[2i-1],M[2i])$}{
    append $M[2i]$ to $X$
  }\Else{
    append $M[2i-1]$ and $M[2i]$ to $Y$
  }
}
run $\textsc{majority}(X)$\;
\lIf{\text{there is no majority in $X$}}{\KwRet{no majority in $M$}}\algseparator
let $\ballcolor(v)$ be the majority with multiplicity $k$ in $X$\;
$\cnt \gets 2k-|X|$\;
\For{$i=1$ \KwSty{to} $|Y|$}{ \label{line:light_loop}
  \If{$\neg \cmp(v,Y[2i-1])$}{ \If{$\neg \cmp(v,Y[2i])$}{ $\cnt \gets \cnt-1$ } }
  \lIf{$\cnt = 0$}{\KwRet{no majority in $M$} }
}
\KwRet{$\ballcolor(v)$ is the majority with multiplicity $(|M|/2+\cnt)$ in $M$}
\end{algorithm}

\begin{lemma}
Algorithm \ref{algo:majority} always returns the correct answer.
\end{lemma}

\begin{proof}
We analyze separately every subprocedure.

$\textsc{balanced}(M)$.
If the majority exists then removing two elements with different colors preserves it.
Hence if the recursive call returns that there is no majority in $X$ then indeed there
is no majority in $M$, and otherwise $\ballcolor(v)$ is the only possible candidate
for the majority in $M$. The remaining part of the subprocedure simply verifies it.

$\textsc{heavy}(M,v)$.
The subprocedure first checks if $\ballcolor(v)$ is the majority. Hence it is enough to
analyze what happens if $\ballcolor(v)$ is not the majority. Then $X$ contains all elements
with other colors. We partition the elements in $X$ into pairs and check which of these pairs
consists of elements with different colors. If the number of elements in all the remaining
pairs is smaller than the number of elements of color $\ballcolor(v)$, then clearly we
can partition all elements in $M$ into disjoint pairs of elements with different colors, hence
indeed there is no majority. Otherwise, we revert to the simple $2n$ algorithm,
which is always correct.

$\textsc{light}(M)$.
Again, if the majority exists then removing two elements with different color preserves it.
Hence we can assume that $\ballcolor(v)$ is the only possible candidate for the
majority. Then, $Y$ consists of pairs of two elements with different
colors. From the recursive call we also know what is the frequency of $\ballcolor(v)$
in $M\setminus Y$. We iterate through the elements of $Y$ and check if their color
is $\ballcolor(v)$. However, if the color of the first element in a pair is $\ballcolor(v)$,
then the second element has a different color. So the subprocedure either correctly
determines the frequency of the majority $\ballcolor(v)$, or find sufficiently many elements
with different colors to conclude that $\ballcolor(v)$ is not the majority.
\end{proof}

\begin{theorem}
\label{lem:upperbound}
Algorithm \ref{algo:majority} w.v.h.p.\ uses at most $\frac{7}{6}n + o(n)$ comparisons on
an input of size~$n$. The expected number of comparisons is also at most $\frac{7}{6}n+o(n)$.
\end{theorem}

\begin{proof}
Let $T(n)$ be a random variable counting the comparisons on the given input of size $n$.
We will inductively prove that $T(n) \le \frac76n + C\cdot n^{9/10}$ for a fixed constant $C$ that is sufficiently large.
In the analysis we will repeatedly invoke Lemmas~\ref{lem:sampling}, \ref{lem:pairs}, \ref{lem:pairs2}, \ref{lem:drawing_until}, \ref{lem:sum_of_squares} and Chernoff bound to bound different quantities.
We will assume that each such the application succeeds.
Since there will be a polynomial number of applications, each on a polynomial number of elements,
this happens w.v.h.p.\ with respect to the size of the input.
We also assume that $n$ is large enough.
Algorithm~\ref{algo:majority} uses at most $\bigo(n^{2\alpha})= \bigo(n^{2/3})$ comparisons
in the sampling stage. We bound the number of subsequent comparisons used by each
subprocedure as follows.

\subparagraph*{{\normalfont $\textsc{balanced}(M)$}.}
We have that $p_1,p_2 = \frac12 \pm \bigo(\varepsilon)$. Thus also $\sum_i p_i^2 =  \frac12 \pm \bigo(\varepsilon)$. By Lemma~\ref{lem:pairs2}, $|X| = (\frac{n}{2}\sum_i p_i^2) \pm \bigo(n^{2/3} \log n)$, thus $|X| = (\frac14 \pm \bigo(\varepsilon)) n$. Also $|Y| = n - 2|X| = (\frac12 \pm \bigo(\varepsilon))n$.

List $Y$ consists of pairs of elements with different colors. Because at most $\bigo(\varepsilon n)$ of
all elements are not of color 1 or 2, there are at most $\bigo(\varepsilon n)$ pairs not
of the form $\{1,2\}$. Since the relative order of elements $Y[2i-1]$ and $Y[2i]$ is random,
for each pair $\{1,2\}$ we pay 1 with probability $1/2$ and pay 2 with probability $1/2$, and for
any other pair we pay always $2$. Thus the total cost incurred by the loop in line~\ref{line:counting1} is
(by Chernoff bound) at most
\[\bigo(\varepsilon n) \cdot 2 + \frac32 |Y|/2 + \bigo(\sqrt{|Y|}\log n) \le \frac38n \pm \bigo(\varepsilon n).\]
Thus the total cost is
\[T(n) \le T\bigl((\tfrac14+\varepsilon)n\bigr) + \tfrac{1}{2}n+\tfrac{3}{8}n+\bigo(\varepsilon n) \le \tfrac76n + \bigo(n^{9/10}) + C\cdot\bigl(\tfrac13n\bigr)^{9/10}\]
and $\frac76n + \bigo(n^{9/10})+C\cdot(\frac13)^{9/10}\cdot n^{9/10} \le \frac76n+ C \cdot n^{9/10}$ for a large enough~$C$.

\subparagraph*{{\normalfont $\textsc{heavy}(M,v)$}.}
If $p_1 > \frac12$, then we terminate in line \ref{line:majority_return} after $n$ comparisons.
Thus we can assume that $p_1\in [0.45-\varepsilon,\frac12]$.
Because by Lemmas~\ref{lem:sampling} and~\ref{lem:sum_of_squares} both
$p_1^2$ and $\sum_i p_i^2$ are estimated within an absolute error of $o(\varepsilon)$,
we have that $p_1^2 - \sum_{i\geq 2}p_i^2 \ge 2\varepsilon - 2o(\varepsilon) \ge \varepsilon$.

We argue that the loop in line~\ref{line:heavy_loop} will eventually find sufficiently many
pairs of elements with different colors, and thus return without falling back to the $2n$ algorithm.
By definition, $|X| = (1-p_1)n$ and initially $k=(1/2-p_1)n$. By Lemma~\ref{lem:pairs2},
after the random shuffle the number $D$ of pairs of elements $(X[2i-1],X[2i])$ with different colors
can be bounded by
\begin{align*}
D &\ge \frac{|X|}{2} - \frac{\sum_{j\geq 2}(p_j n)^2}{2|X|} - \bigo\bigl(|X|^{2/3} \log |X|\bigr) \ge \\
&\ge \frac{1-p_1}2 n - \frac{p_1^2-\varepsilon}{2(1-p_1)}n - o(\varepsilon n) \ge \\
&\ge \frac{1-2p_1}{2(1-p_1)} n + \frac{\varepsilon}{2}n - o(\varepsilon n)\ge \\
&\ge \bigl(\tfrac12 - p_1\bigr)n,
\end{align*}
thus indeed there are sufficiently many pairs. Hence, because the pairs are being considered in
a random order, the total cost can be bounded using Lemma~\ref{lem:drawing_until} by
\begin{align*}
T(n) &\le n+\frac{|X|}{D} \biggl(\frac12-p_1\biggr)n  +\bigo\biggl(\frac{|X|}{\sqrt{D}} \log |X|\biggr) \le \\
&\le  n + \frac{(1-p_1)^2}{2}n + \bigo\biggl(n / \sqrt{\frac{\varepsilon}3 n} \log n\biggr) \le \\
&\le (1+0.55^2/2)n + \bigo(\varepsilon n) + \bigo\biggl(\sqrt{\frac{n}{\varepsilon}}\log n\biggr) = \\
& = 1.15125 n + \bigo(n^{9/10}),
\end{align*}
where we used $D \ge \frac{\varepsilon}{2}-o(\epsilon n) \ge \frac{\varepsilon}{3}n$
for a large enough~$n$.

\subparagraph*{{\normalfont $\textsc{light}(M)$}.}
We start with bounding $|X|$ and $|Y|$. By Lemma~\ref{lem:pairs2},
$|X| = \frac{n}{2} \sum_i p_i^2 \pm \bigo(n^{2/3} \log n)$, and by Lemma~\ref{lem:pairs} there
are $\frac{n}{2} p_1^2 \pm \bigo(n^{1/2} \log n)$ elements from $A_1$ in $X$, thus there are
$n (p_1 -p_1^2) \pm \bigo(n^{1/2} \log n)$ of elements from $A_1$ in $Y$ (each paired with
a non-$A_1$ element).

We know that either $p_1 \le 0.45+\varepsilon$ or $p_1^2 - \sum_{i\geq 2}(p_i^2) \le \varepsilon$.
If there is no majority in $X$, then $p_1 \le \frac{1}{2}$ and the total cost is bounded by
\[T(n) \le \frac{n}{2} + T(|X|) \le \frac{n}{2} + \frac{7}{6} |X| + C \cdot |X|^{9/10},\]
which, because $|X| \leq \frac{n}{4}+\bigo(n^{2/3}\log n)$, is less than $\frac{19}{24}n+o(n)$.
Hence we can assume that there is a majority in $X$. In such case, $\cnt$ is set to
\[c = \frac{n}{2} \biggl(p_1^2 - \sum_{i\geq 2}p_i^2\biggr) \pm \bigo(n^{2/3} \log n).\]
We denote by $I$ the total number of iterations of the loop in line~\ref{line:light_loop}.
By Lemma~\ref{lem:drawing_until}
\[I \le \frac{\frac12 |Y|}{\frac12|Y| - |A_1 \cap Y|}  \cdot c + \bigo(E), \]
where $E = |Y| / \sqrt{\frac12|Y| - |A_1 \cap Y|}$.
Substituting $S = \sum_{i\geq 2}p_i^2$, by Lemma~\ref{lem:pairs2} we have
\begin{align*}
|Y| &= \bigl(1 - p_1^2 - S \pm \bigo(n^{-1/3}\log n)\bigr)n, \\
|Y| - 2|A_1 \cap Y| &= \bigl((1 - p_1)^2 - S \pm \bigo(n^{-1/3}\log n)\bigr)n, \\
c &=\frac12\bigl( p_1^2 - S \pm \bigo(n^{-1/3}\log n)\bigr)n.
\end{align*}
Since $p_1 \le \frac12$ and $p_2 \le \frac12-3\varepsilon$ (as for a larger $p_2$ the sampled
$q_2$ would be sufficiently large for other subprocedure to be used),
we have \[(1-p_1)^2-S - o(\varepsilon) \ge \bigl(\tfrac12\bigr)^2-\bigl(\tfrac12-3\varepsilon\bigr)^2-(3\varepsilon)^2 - o(\varepsilon) = 3\varepsilon - 18\varepsilon^2 - o(\varepsilon) \ge 2\varepsilon.\] Thus  $E \le \sqrt{\frac{n}{2\varepsilon}}$. Now, since $|Y| = \Theta(n)$ we can bound $I$ from above by
\begin{align*}
I & \le \frac{|Y|}{|Y|-2|A_1\cap Y|}\cdot \frac12(p_1^2-S)n + \bigo(1/\varepsilon)\cdot \bigo(n^{2/3}\log n) + \bigo(E) \le \\
& \le \frac{1}{2}\frac{1-p_1^2-S+\bigo(n^{-1/3}\log n)}{(1-p_1)^2-S-\bigo(n^{-1/3}\log n)}(p_1^2-S)n + \bigo\biggl(\frac{n^{2/3}\log n}{\varepsilon}\biggr) + \bigo\biggl(\sqrt{\frac{n}{4\varepsilon}}\biggr)\le \\
& \le \frac{1}{2}\frac{1-p_1^2-S}{(1-p_1)^2-S-\bigo(n^{-1/3}\log n)}(p_1^2-S)n +\frac{\bigo(n^{2/3}\log n)}{2\varepsilon} + \bigo(n^{23/30}\log n),
\end{align*}
which, because $(1-p_1)^2-S$ is sufficiently large, can be bounded by
\begin{align*}
I & \le \frac{1}{2}\frac{1-p_1^2-S}{(1-p_1)^2-S}(p_1^2-S)n\cdot \biggl(1+\frac{\bigo(n^{-1/3}\log n)}{(1-p_1)^2-S}\biggr)  + \bigo(n^{23/30}\log n) \le \\
& \le \frac{1}{2}\frac{1-p_1^2-S}{(1-p_1)^2-S}(p_1^2-S)n\cdot \biggl(1+\frac{\bigo(n^{-1/3}\log n)}{2\varepsilon}\biggr)  + \bigo(n^{23/30}\log n) \le \\
& \le \frac{1}{2}\frac{1-p_1^2-S}{(1-p_1)^2-S}(p_1^2-S)n+\frac{\bigo(n^{2/3}\log n)}{8\varepsilon^2}  + \bigo(n^{23/30}\log n) \le \\
&\le\frac{1}{2}(1-p_1^2-S)\frac{p_1^2-S}{(1-p_1)^2-S}n + \bigo(n^{26/30}\log n).
\end{align*}

For each of $c$ iterations we pay $2$, and for each of the remaining $I-c$ iterations we pay
only $\frac{3}{2}$ in expectation (for each iteration independently). Thus, by Chernoff bound
the total cost is
\begin{align*}
T(n) &\le \frac12n + T(|X|)+\frac32(I-c)+\bigo(\sqrt{I-c}\log (I-c))+2c \le \\[6pt]
& \le \frac12n + \frac{7}{12}n(p_1^2+S) + \frac{3}{4}(1-p_1^2-S)\frac{p_1^2-S}{(1-p_1)^2-S}n  +  \frac14(p_1^2-S)n + \bigo(n^{26/30}\log n)  + C \cdot |X|^{9/10} = \\[6pt]
& = \frac{n}{2}\biggl(1+\frac76(p_1^2+S)+\frac32 (p_1^2-S)\frac{1-p_1^2-S-((1-p_1)^2-S)}{(1-p_1)^2-S} + \frac32(p_1^2-S) + \frac12(p_1^2-S)\biggr)+ \bigo( n^{9/10}) =\\[6pt]
&= \frac{n}{2}\biggl(1+\frac{19}6p_1^2-\frac56S+3 (p_1^2-S)\frac{p_1-p_1^2}{(1-p_1)^2-S} \biggr)+ \bigo( n^{9/10}).
\end{align*}

We reason that, for a fixed $p_1$, the quantity \[1+\frac{19}{6}p_1^2-\frac{5}{6}S + 3 (p_1^2-S) \frac{p_1(1-p_1)}{(1-p_1)^2-S}\]  is a decreasing function of $S$, since $p_1^2 \le (1-p_1)^2$. Now we consider two cases. If $p_1^2 - S \le \varepsilon$ then  simplifying above estimation either with $p_1^2-S \le 0$, or, since $(1-p_1)^2-S \ge 2\varepsilon$, with $0 \le \frac{p_1^2-S}{(1-p_1)^2-S} \le \frac12$, we get either
\[T(n) \le \frac{n}{2}\left(1+\frac73p_1^2 \right) + \bigo(n^{9/10}), \]
which, since $p_1 \le \frac12$, is bounded from above by $\frac{19}{24}n + o(n)$, or
\[T(n) \le \frac{n}{2}\left(1+\frac73p_1^2 + \frac32p_1-\frac32p_1^2\right) + \bigo( n^{9/10}), \]
which is bounded from above by $\frac{47}{48}n+o(n)$.

Otherwise, $p_1 \le 0.45 + \varepsilon$ and substituting $S=0$ (since the cost is decreasing in $S$) we obtain
\[T(n) \le \frac{n}{2}\biggl(1+\frac{19}{6}p_1^2 + 3 \frac{p_1^3}{1-p_1}\biggr) + \bigo(n^{9/10}) = 1.06915 n + \bigo(n^{9/10}).\]

\subparagraph*{Wrapping up.}
We see that in each subprocedure, the number of comparisons is bounded by $\frac76n+C\cdot n^{9/10}$. 
Each subprocedure makes at most one recursive call, where the size of the input is reduced by at least
a factor of 2. Thus the worst-case number of comparison is always bounded by $\bigo(n)$.
Recall that the bound on the number of comparisons used by every recursive call holds
w.v.h.p.\ with respected to the size of the input to the call.
Eventually, the size of the input might become very small, and then w.v.h.p.\ with respect to the size
of the input is no longer w.v.h.p.\ with respect to the original $n$. However, as soon as this size
decreases to, say, $n^{0.1}$, the number of comparisons is $\bigo(n)$ irrespectively of the
random choices made by the algorithm. Thus w.v.h.p.\ the number of comparisons is at most
$\frac76n+\bigo(n^{9/10})$, and the expected number of comparisons is also bounded by
$\frac{7}{6}n + \bigo(n^{9/10})$.
\end{proof}

\section{Lower bound}

We consider Las Vegas algorithms. That is, the algorithm must always correctly determine whether
a majority element exists.  We will prove that the expected number of comparisons used by such
an algorithm is at least $c\cdot n-o(n)$, for some constant $c > 1$.
By Yao's principle, it is sufficient to construct a distribution on the inputs,
such that the expected number of comparisons used by any deterministic algorithms run on an input
chosen from the distribution is at least $c\cdot n-o(n)$.
The distribution is that with probability $\frac{1}{n}$ every ball has a color chosen uniformly at random
from a set of $n$ colors. With probability $1-\frac{1}{n}$ every ball is black or white,
with both possibilities equally probable. We fix a correct deterministic algorithm $\mathcal{A}$ and
analyze its behavior on an input chosen from the distribution. As a warm-up, we
first prove that $\mathcal{A}$ needs $n-o(n)$ comparisons in expectation on such input.

\subsection{A lower bound of \texorpdfstring{\boldmath $n - o(n)$}{n - o(n)}}

In every step $\mathcal{A}$ compares two balls, thus we can describe its current knowledge by defining
an appropriate graph as follows. Every node corresponds to a ball. Two nodes are connected with
a \emph{negative edge} if the corresponding balls have been compared and found out to have
different colors. Two nodes
are connected with a \emph{positive edge} if the corresponding balls are known to have the same
colors under the assumption that every ball is either black or white (either because they have been
directly compared and found to have the same color, or because such knowledge has been indirectly
inferred from the assumption).
After every step of the algorithm the graph consists of a number of components $C_1,C_2,\ldots$.
Every components is partitioned into two parts $C_i = A _i \cupdot B_i$, such that both $A_i$ and
$B_i$ are connected components in the graph containing only positive edges and there is at least one
(possibly more than one) negative edge between $A_i$ and $B_i$. There are no other edges in the graph.
Now we describe how the graph changes after $\mathcal{A}$ compares two balls $x\in C_i$ and $y\in C_j$
assuming that every ball is either black or white. 
If $i=j$ then the result of the comparison is already determined by the previous comparisons and
the graph does not change. Otherwise, $i\neq j$ and assume by symmetry that $x\in A_i, y\in A_j$.
The following two possibilities are equally probable:
\begin{enumerate}
\item $\ballcolor(x) = \ballcolor(y)$, then we merge both components into a new component
$C = A \cupdot B$, where $A = A_i \cupdot A_j$ and $B = B_i \cupdot B_j$ by creating new
positive edges $(x,y)$ and $(x',y')$ for some $x'\in B_i, y'\in B_j$ (if $B_i ,B_j\neq\emptyset$).
\item $\ballcolor(x) \neq \ballcolor(y)$,  then we merge both components into a new component
$C = A  \cupdot B$, where $A = A_i \cupdot B_j$ and $B = B_i \cupdot A_j$ by creating 
new positive edges $(x,y')$ for some $y'\in B_j$ (if $B_j\neq \emptyset$) and
$(x',y)$ for some $x'\in B_i$ (if $B_i\neq\emptyset$). We also create a new
negative edge $(x,y)$. Here we
crucially use the assumption that every ball is either black or white.
\end{enumerate}
The graph exactly captures the knowledge of $\mathcal{A}$ about a binary input. 

Any binary input contains a majority and $\mathcal{A}$ must report so. However, because with
very small probability the input is arbitrary, this requires some work due to the following
lemma.

\begin{lemma}
\label{lem:terminate}
If $\mathcal{A}$ reports that a binary input contains a majority element, then the graph contains a
component $C = A \cupdot B$ such that $|A| > \frac{n}{2}$ or $|B| > \frac{n}{2}$.
\end{lemma}

\begin{proof}
Assume otherwise, that is, $\mathcal{A}$ reports that a binary input contains a majority element
even though both parts of every component are of size less than $\frac{n}{2}$. Construct another
input by choosing, for every component $C = A \cupdot B$, two fresh colors $c_A$ and $c_B$ and
setting $\ballcolor(x) = c_A$ for every $x\in A$, $\ballcolor(y) = c_B$ for every $y \in B$.
Every comparison performed by $\mathcal{A}$ is an edge of the graph, so its behavior on the new input
is exactly the same as on the original binary input. Hence $\mathcal{A}$ reports that there is
a majority element, while the frequency of every color in the new input is less than $\frac{n}{2}$,
which is a contradiction.
\end{proof}

From now on we consider only binary inputs. If we can prove that the expected number of comparisons
used by $\mathcal{A}$ on such input is $n-o(n)$, then the expected number of comparisons 
on an input chosen from our distribution is also $n-o(n)$.
Because every comparison decreases the number of components by one, it is sufficient to argue
that the expected size of some component when $\mathcal{A}$ reports that there is a majority
is $n-o(n)$. We already know that there must exist a component $C = A \cupdot B$ such that
(by symmetry) $|A| > n/2$. We will argue that $|B|$ must also be large.
To this end, define \emph{balance} of a component $C_i = A_i \cupdot B_i$ as
$\balance(C_i) = (|A_i|-|B_i|)^2$, and the total balance as $\sum_i \balance(C_i)$.
By considering the situation before and after a single comparison, we obtain the following.

\begin{lemma}
\label{lem:sum}
The expected total balance at termination of algorithm $\mathcal{A}$ is $n$.
\end{lemma}

\begin{proof}
In the very beginning the total balance is $n$ because every component is a singleton. Recall that
when $\mathcal{A}$ compares two balls $x\in C_i$ and $y\in C_j$ then $C_i$ and $C_j$ are
merged into a new component $C$. It is easy to verify that if $\balance(C_i)=b^2_i$ and
$\balance(C_j)=b^2_j$ then either $\balance(C) = (b_i + b_j)^2 $ or $\balance(C) = (b_i - b_j)^2$,
with both possibilities equally probable. Hence the expected total balance after the comparison
is equal to the previous total balance increased by
$ -b^2_i - b^2_j + \frac12((b_i - b_j)^2 + (b_i + b_j)^2) = 0$. 
Consequently, total balance is a martingale and is preserved in expectation with respect to arbitrarily branching computation.
\end{proof}

Total balance when $\mathcal{A}$ reports a majority is a random variable with expected value $n$.
By Markov's inequality, with probability $1-1/n^{1/3}$ its value is at most $n^{4/3}$, which
implies that for any component $C_i = A_i \cupdot B_i$, we have $\balance(C_i) \leq n^{4/3}$. If we apply this inequality to the component $C = A \cupdot B$ with $|A| > n/2$, we obtain $|B| \geq n/2 - n^{2/3}$.
Hence with probability $1-1/n^{1/3}$ there is a component with at least $n-n^{2/3}$
nodes,
which means that the algorithm must have performed at least $n-n^{2/3} - 1$ comparisons.
Therefore the expected number of comparisons is at least $(1-1/n^{1/3})(n-n^{2/3}-1) = n - o(n)$.

\subsection{A stronger lower bound}
\label{section:stronger_lowerbound}

To obtain a stronger lower bound, we extend the definition of the graph that captures the current knowledge
of~$\mathcal{A}$. Now a positive edge can be \emph{verified} or \emph{non-verified}. A~verified
positive edge $(x,y)$ is created only after comparing two balls $x$ and $y$ such that $\ballcolor(x)=\ballcolor(y)$.
All other positive edges are non-verified. The algorithm can also turn a non-verified positive edge
$(x,y)$ into a verified positive edge by comparing $x$ and $y$. By the same reasoning as in
Lemma~\ref{lem:terminate} we obtain the following.

\begin{lemma}
\label{lem:terminate2}
If $\mathcal{A}$ reports that a binary input contains a majority element, then the graph consisting
of all verified positive edges contains a connected component with at least $\frac{n}{2}$ nodes.
\end{lemma}

Now the goal is to construct a large component in the graph that consists of all verified positive
edges, so it makes sense for $\mathcal{A}$ to compare two balls from the same component.
However, without loss of generality, such comparisons are executed after having identified a
large component in the graph consisting of all positive edges. Then, $\mathcal{A}$ asks sufficiently many
queries of the form $(x,y)$, where $(x,y)$ is a non-verified edge from the identified component.
In other words, $\mathcal{A}$ first isolates a candidate for a majority, and then makes sure that
all inferred equalities really hold, which is necessary because with very small probability the input
is not binary. This allows us to bound the total number of comparisons from below as follows. We define that a \emph{majority edge} is an edge between two nodes of the majority color.

\begin{lemma}
\label{lem:verify}
The expected number of comparisons used by $\mathcal{A}$ on a binary input is at least
$n-o(n)$ plus the expected number of non-verified majority edges.
\end{lemma}

\begin{proof}
Recall that if there exists a component $C=A\cupdot B$ with $|A|>n/2$
then with probability $1-1/n^{1/3}$ we also have $|B| \geq n/2-n^{2/3}$.
Set $A$ consists of nodes of the majority color, although possibly not all nodes of the majority color are there. 
However, because $B$ is large, there are at most $n^{2/3}$ nodes of the majority color
outside of $A$. Also, because we consider binary inputs chosen uniformly at random,
by Chernoff bound $|A| \leq n/2+\bigo(\sqrt{n\log n})$ with probability $1-1/n$.

The expected number of comparisons used by $\mathcal{A}$ to construct
a component $C=A\cupdot B$ such that $|A|> n/2$ is at least $n-n^{2/3}-1$.
Then, $\mathcal{A}$ needs to verify sufficiently many non-verified edges inside $A$ to obtain
a connected component of size $n/2$ in the graph that consists of verified positive
edges. By construction, there are no cycles in the graph that consists of positive edges.
Hence with probability $1-1/{n^{1/3}}-1/n$ there will be no more than
$n^{2/3}+\bigo(\sqrt{n\log n})$ non-verified positive edges between nodes outside of $B$
when $\mathcal{A}$ reports a majority. Consequently, the additional verification cost
is the expected number of non-verified majority edges minus
$n^{2/3}+\bigo(\sqrt{n\log n})=o(n)$.
\end{proof}

In the remaining part of this section we analyze the expected number of non-verified
majority edges constructed during the execution of the algorithm.
We show that this is at least $(c-1)n-o(n)$ for some $c > 1$. Then, Lemma~\ref{lem:verify} implies
the claimed lower bound.

A component $C=A\cupdot B$ is called \emph{monochromatic} when $A=\emptyset$ or
$B=\emptyset$ (by symmetry, we will assume the latter) and \emph{dichromatic} otherwise.
With probability $1-1/n^{1/3}$, when $\mathcal{A}$ reports a majority there is
one large dichromatic component with at least $n-n^{2/3}$ nodes, and hence the total number of
components is at most $n^{2/3}+1$. It is convenient to interpret the execution of $\mathcal{A}$
as a process of eliminating components by merging two components into one.
Each such merge might create a new non-verified edge. We define that the \emph{cost} of such a non-verified edge is the probability that it is a majority edge.
We want to argue that because all but $n^{2/3}$ components will be eventually eliminated,
the total cost of all non-verified edges that we create is $(c-1)n-o(n)$.

We analyze in more detail the merging process in terms of mono- and dichromatic components.
Let $\predict_k$ be the random variable denoting the probability that, after $k$ steps  of $\mathcal{A}$, a node from the larger
part of a component is of the majority color. It is rather difficult to calculate $\predict_k$
exactly, so we will use a crude upper bound instead. An important property of the
upper bound will be that it is nondecreasing in $k$.
When $\mathcal{A}$ compares two balls $x\in C_i$ and $y\in C_j$ with $i\neq j$ to obtain
a new component $C=A\cupdot B$ there are three possible cases:
\begin{enumerate}
\item $C_i$ and $C_j$ are monochromatic. Then with probability $\frac{1}{2}$ the new
component $C$ is also monochromatic, and with probability $\frac{1}{2}$ it is dichromatic.
\item $C_i$ is dichromatic and $C_j$ is monochromatic. The new component is dichromatic.
With probability $\frac{1}{2}$ we have a new non-verified edge, and with probability at least $\frac{1}{2}(1-\predict_k)$ we have a new non-verified majority edge.
\item $C_i$ and $C_j$ are dichromatic. Then with probability $\frac{1}{2}$ we create a new non-verified
edge inside both $A$ and $B$, and one of them is a majority edge.
\end{enumerate}

We analyze the expected total cost of all non-verified edges when only one component remains.
When $\mathcal{A}$ reports a majority up to $n^{2/3}$ components might remain, but this
changes only the lower order terms of the bound.

\begin{lemma}
\label{lem:nonverified}
The expected total cost of all non-verified edges when only one component remains is at least
$\sum_{k=1}^{2n/3} \E\bigl[\min\bigl(\tfrac{1}{6},\tfrac{1}{2}(1-\predict_k)\bigr)\bigr]$.
\end{lemma}

\begin{proof}
We start with $n$ components and need to eliminate all but at most one
of them. To each component we associate credit, $\frac{1}{2}$ to each dichromatic and $\frac{1}{6}$ to each monochromatic one. The algorithm can collect the credit from both of the components it merges, but it has to pay for credit of newly created one. Additionally algorithm has to pay for any non-verified majority edge created by merging.

In every step we have three possibilities:
\begin{enumerate}
\item Merge two monochromatic components into one. With probability $\frac{1}{2}$ the new component is dichromatic, and with probability $\frac{1}{2}$ the new component is monochromatic. Thus the expected amortized cost for this step is $0$.
\item Merge a monochromatic components with a dichromatic component. Then the total
number of monochromatic components decreases by $1$ and we add with probability at least $\frac{1}{2}(1-\predict_k)$ a non-verified majority edge. The expected amortized cost for this step is $\frac{1}{2}(1-\predict_k) - \frac16$.
\item Merge two dichromatic components while adding with probability $\frac{1}{2}$ a non-verified majority edge. The expected amortized cost for this step is $0$.
\end{enumerate}
In total algorithm has to pay for initial credits and for each step, making the total expected cost at least
\[\frac{n}6 + \sum_{k=1}^{n-1} \E\bigl[\min\bigl(0,\tfrac{1}{2}(1-\predict_k) - \frac16\bigr)\bigr] \ge  \sum_{k=1}^{2/3n} \E\bigl[\min\bigl(\tfrac{1}{6},\tfrac{1}{2}(1-\predict_k)\bigr) \bigr].  \qedhere \]
\end{proof}
We note that by truncating the sum at $\frac23n$ we do not lose any cost estimation, as for $k \ge \frac23n$ our estimation for $\predict_k$ gives $1$. 

Now we focus deriving an upper bound for the expression obtained in Lemma~\ref{lem:nonverified}.
To bound $\predict_k$ we use an approach due to Christofides~\cite{Christofides}.
At any given step $k$ we will look at all components with a nonzero balance. Specifically, we introduce two new random variables: $M_k$ being the largest balance of a component, and $N_k$ being the number of components with a nonzero balance. Since at each step, $N_k$ is decreased in expectation at most by $\frac32$, we have
$\E[N_k] \ge n - \frac32(k-1)$,
and w.v.h.p., by Chernoff bound
$N_k \ge n - \frac32k - \bigo(\sqrt{k} \log n)$.

Since by Lemma~\ref{lem:sum} the expected sum of balances is $n$, and each nonzero component contributes at least $1$ to the sum, we have
$\E[M_k] \le n - \E[N_k-1] = \frac32k-\frac12$.

Now to proceed, for a component $C_i = A_i \cupdot B_i$ we define $\delta_i = \left||A_i| - |B_i|\right|$, a positive value such that $\delta_i^2 = \balance(C_i)$. Thus, at any given step $k$, the algorithm observes the nonzero values $\delta_1,\delta_2,\ldots,\delta_{N_k}$. Without loss of generality we can narrow our focus on a component~$C_1$. We are interested in bounding the probability
\[ \Pr(A_1 \text{ in majority}) = \Pr(  \delta_1 + \varepsilon_2 \delta_2 \ldots + \varepsilon_{N_k} \delta_{N_k} \ge 0 ) = \tfrac12 + \tfrac12 \Pr( \varepsilon_2 \delta_2 \ldots + \varepsilon_{N_k} \delta_{N_k} \in [-\delta_1,\delta_1]),\]
where $\epsilon_2,\epsilon_3,\ldots,\epsilon_{N_k} \in \{-1,1\}$ are drawn independently and uniformly at random.
By a result of Erd\H{o}s~\cite{erdos1945}, if $\delta_2, \ldots, \delta_{N_k} \ge 1$ then the above is maximized for $\delta_2=\ldots=\delta_{N_k}=1$.

We now approximate binomial distribution using the symmetric case of de Moivre--Laplace Theorem. Recall that 
\[\Phi(x) = \frac{1}{\sqrt{2\pi}} \int_{-\infty}^x \! e^{-t^2/2} \, \mathrm{d}t\]
is the \emph{cumulative distribution function} of the \emph{normal distribution}.
\begin{theorem}[De Moivre--Laplace]
Let $S_n$ be the number of successes in $n$ independent coin flips. Then
\[ \Pr\Bigl(\frac{n}{2} + x_1\sqrt{n} \le S_n \le \frac{n}{2} + x_2 \sqrt{n}\Bigr) \sim \Phi(2x_2) - \Phi(2x_1). \]
\end{theorem}
In our case we are interested in $N_k-1$ coin flips and the number of successes in the range $[(N_k-1)/2-\delta_1/2,\,(N_k-1)/2+\delta_1/2]$. Thus probability that 1 is the majority can be bounded from above by
\[ \Pr(A_1 \text{ is the majority}) \le \frac12\biggl(\Phi\biggl(\frac{\delta_1}{\sqrt{N_k-1}}\biggr) - \Phi\biggl(-\frac{\delta_1}{\sqrt{N_k-1}}\biggr)\biggr)+\frac12 = \Phi\biggl(\frac{\delta_1}{\sqrt{N_k-1}}\biggr).\]
Because $M_k$ is the largest balance of a component, $\delta_1,\delta_2,\ldots,\delta_{N_k}$ are bounded from above by $\sqrt{M_k}$. Additionally, w.v.h.p.\ $N_k \ge n-\frac32k - \bigo(\sqrt{n} \log n)$, thus
\[\predict_k \le \Phi\left(\sqrt{\frac{M_k}{n-\frac32k - \bigo(\sqrt{n} \log n)}}\right).\]

Since $\Phi(\sqrt{x}/\text{const})$ is a concave function, we can apply expected value, and get
\[ \E[\predict_k] \le \Phi\left(\sqrt{\frac{\E[M_k]}{n-\frac32k - \bigo(\sqrt{n} \log n)}}\right) \sim \Phi\left(\sqrt{\frac{\frac32k}{n-\frac32k}}\right). \]

Now we are ready to bound the sum from Lemma~\ref{lem:nonverified}. Using the linearity
of expectation and inequality $\min(\frac{1}{6},\frac{1}{2}x) \geq \frac{1}{6}x$ for $x\in[0,1]$ we obtain:
\[ \E\left[\sum_{k=1}^{2n/3} \min\biggl(\frac16,\frac12(1-\predict_k)\biggr)\right] = \sum_{k=1}^{2n/3} \E\biggr[\min\biggl(\frac16,\frac12(1-\predict_k)\biggr)\biggr] \ge \]
\[ \ge \sum_{k=1}^{2n/3} \frac16 (1-\E[\predict_k]) \ge  n\cdot \int_{0}^{2/3} \! \frac16\left(1-\Phi\left(\sqrt{\frac{\frac32x}{1-\frac32x}}\right)\right) \, \mathrm{d}x - o(n).\]
Finally, we calculate \[1 + \int_{0}^{2/3} \! \frac16\left(1-\Phi\left(\sqrt{\frac{\frac32x}{1-\frac32x}}\right)\right) \, \mathrm{d}x \approx 1.0191289.\]

\begin{theorem}
Any algorithm that reports majority exactly requires in expectation at least $1.019n$ comparisons.
\end{theorem}

\section{Conclusions}
We have presented a Las Vegas algorithm for finding a majority color ball using, with high probability,
$\frac{7}{6}n+o(n)$ comparisons. We have also shown that the expected number of comparisons
needs to be at least $1.019n$. We believe that a more careful application of our methods might
slightly increase the lower bound, but achieving $\frac{7}{6}n$, which we believe to be the answer, requires
a new approach. Another interesting question is to consider Monte Carlo algorithms.

\section*{Acknowledgments}
Most of the work has been done while PU was affiliated to Aalto University, Finland.

\bibliographystyle{abbrv}
\bibliography{bib}

\begin{thebibliography}{10}

\bibitem{AignerMM05}
M.~Aigner, G.~D. Marco, and M.~Montangero.
\newblock The plurality problem with three colors and more.
\newblock {\em Theor. Comput. Sci.}, 337(1-3):319--330, 2005.

\bibitem{AlonsoLowerbound}
L.~Alonso, E.~M. Reingold, and R.~Schott.
\newblock Determining the majority.
\newblock {\em Inf. Process. Lett.}, 47(5):253--255, 1993.

\bibitem{AlonsoAverage}
L.~Alonso, E.~M. Reingold, and R.~Schott.
\newblock The average-case complexity of determining the majority.
\newblock {\em {SIAM} J. Comput.}, 26(1):1--14, 1997.

\bibitem{MJRTY}
R.~S. Boyer and J.~S. Moore.
\newblock {MJRTY:} {A} fast majority vote algorithm.
\newblock In {\em Automated Reasoning: Essays in Honor of Woody Bledsoe}, pages
  105--118, 1991.

\bibitem{Christofides}
D.~Christofides.
\newblock On randomized algorithms for the majority problem.
\newblock {\em Discrete Applied Mathematics}, 157(7):1481--1485, 2009.

\bibitem{oblivious}
F.~R.~K. Chung, R.~L. Graham, J.~Mao, and A.~C. Yao.
\newblock Oblivious and adaptive strategies for the majority and plurality
  problems.
\newblock {\em Algorithmica}, 48(2):147--157, 2007.

\bibitem{chvatal}
V.~Chv\'atal.
\newblock The tail of the hypergeometric distribution.
\newblock {\em Discrete Mathematics}, 25(3):285--287, 1979.

\bibitem{DorZ99}
D.~Dor and U.~Zwick.
\newblock Selecting the median.
\newblock {\em {SIAM} J. Comput.}, 28(5):1722--1758, 1999.

\bibitem{DorZ01}
D.~Dor and U.~Zwick.
\newblock Median selection requires $(2+\epsilon)n$ comparisons.
\newblock {\em {SIAM} J. Discrete Math.}, 14(3):312--325, 2001.

\bibitem{Eppstein}
D.~Eppstein and D.~S. Hirschberg.
\newblock From discrepancy to majority.
\newblock In {\em {LATIN}}, volume 9644 of {\em Lecture Notes in Computer
  Science}, pages 390--402. Springer, 2016.

\bibitem{erdos1945}
P.~Erd{\H o}s.
\newblock On a lemma of {Littlewood and Offord}.
\newblock {\em Bull. Amer. Math. Soc.}, 51(12):898--902, 12 1945.

\bibitem{FischerSalzberg}
M.~Fischer and S.~Salzberg.
\newblock {Finding a majority among $n$ votes: solution to problem 81-5}.
\newblock {\em Journal of Algorithms}, 1982.

\bibitem{GerbnerKPP13}
D.~Gerbner, G.~O.~H. Katona, D.~P{\'{a}}lv{\"{o}}lgyi, and B.~Patk{\'{o}}s.
\newblock Majority and plurality problems.
\newblock {\em Discrete Applied Mathematics}, 161(6):813--818, 2013.

\bibitem{KralST08}
D.~Kr{\'{a}}l, J.~Sgall, and T.~Tich{\'{y}}.
\newblock Randomized strategies for the plurality problem.
\newblock {\em Discrete Applied Mathematics}, 156(17):3305--3311, 2008.

\bibitem{MarcoK15}
G.~D. Marco and E.~Kranakis.
\newblock Searching for majority with k-tuple queries.
\newblock {\em Discrete Math., Alg. and Appl.}, 7(2), 2015.

\bibitem{MarcoPelc}
G.~D. Marco and A.~Pelc.
\newblock Randomized algorithms for determining the majority on graphs.
\newblock {\em Combinatorics, Probability {\&} Computing}, 15(6):823--834,
  2006.

\bibitem{paterson1996progress}
M.~Paterson.
\newblock Progress in selection.
\newblock In {\em Algorithm Theory--SWAT'96}, pages 368--379. Springer, 1996.

\bibitem{SaksWerman}
M.~E. Saks and M.~Werman.
\newblock On computing majority by comparisons.
\newblock {\em Combinatorica}, 11(4):383--387, 1991.

\bibitem{VizerGKPPW15}
M.~Vizer, D.~Gerbner, B.~Keszegh, D.~P{\'{a}}lv{\"{o}}lgyi, B.~Patk{\'{o}}s,
  and G.~Wiener.
\newblock Finding a majority ball with majority answers.
\newblock {\em Electronic Notes in Discrete Mathematics}, 49:345--351, 2015.

\bibitem{Wiener}
G.~Wiener.
\newblock Search for a majority element.
\newblock {\em Journal of Statistical Planning and Inference}, 100(2):313--318,
  2002.

\end{thebibliography}

\end{document}